\documentclass{article}

\usepackage[main, final]{neurips_2026}

\usepackage[utf8]{inputenc} 
\usepackage[T1]{fontenc}    
\usepackage[pdftex,colorlinks,linkcolor=blue,citecolor=blue,filecolor=blue,urlcolor=blue]{hyperref}            
\usepackage{booktabs}       
\usepackage{amsfonts}       
\usepackage{nicefrac}       
\usepackage{microtype}      
\usepackage{xcolor}         

\usepackage{amsmath}
\usepackage{amssymb}
\usepackage{mathtools}
\usepackage{amsthm}

\usepackage{graphicx}
\usepackage{subcaption}

\usepackage{cleveref}

\usepackage{aliascnt}

\newtheorem{theorem}{Theorem}[section]
\crefname{theorem}{Theorem}{Theorems}

\newaliascnt{lemma}{theorem}
\newtheorem{lemma}[lemma]{Lemma}
\aliascntresetthe{lemma}
\crefname{lemma}{Lemma}{Lemmas}

\newaliascnt{proposition}{theorem}

\aliascntresetthe{proposition}
\crefname{proposition}{Proposition}{Propositions}

\newaliascnt{corollary}{theorem}
\newtheorem{corollary}[corollary]{Corollary}
\aliascntresetthe{corollary}
\crefname{corollary}{Corollary}{Corollaries}

\newaliascnt{fact}{theorem}

\aliascntresetthe{fact}
\crefname{fact}{Fact}{Facts}

\newaliascnt{definition}{theorem}
\newtheorem{definition}[definition]{Definition}
\aliascntresetthe{definition}
\crefname{definition}{Definition}{Definitions}

\newaliascnt{remark}{theorem}

\aliascntresetthe{remark}
\crefname{remark}{Remark}{Remarks}

\newaliascnt{conjecture}{theorem}

\aliascntresetthe{conjecture}
\crefname{conjecture}{Conjecture}{Conjectures}

\newaliascnt{claim}{theorem}

\aliascntresetthe{claim}
\crefname{claim}{Claim}{Claims}

\newaliascnt{question}{theorem}

\aliascntresetthe{question}
\crefname{question}{Question}{Questions}

\newaliascnt{exercise}{theorem}

\aliascntresetthe{exercise}
\crefname{exercise}{Exercise}{Exercises}

\newaliascnt{example}{theorem}

\aliascntresetthe{example}
\crefname{example}{Example}{Examples}

\newaliascnt{notation}{theorem}

\aliascntresetthe{notation}
\crefname{notation}{Notation}{Notations}

\newaliascnt{problem}{theorem}

\aliascntresetthe{problem}
\crefname{problem}{Problem}{Problems}

\newcommand\numberthis{\addtocounter{equation}{1}\tag{\theequation}}

\newenvironment{CompactEnumerate}{
\begin{list}{\arabic{enumi}.}{%
\usecounter{enumi}
\setlength{\leftmargin}{10pt}
\setlength{\itemindent}{2pt}
\setlength{\topsep}{-1pt}
\setlength{\itemsep}{-1pt}
}}
{\end{list}}

\newenvironment{CompactItemize}{
\begin{list}{$\bullet$}{%
\setlength{\leftmargin}{10pt}
\setlength{\itemindent}{2pt}
\setlength{\topsep}{2pt}
\setlength{\itemsep}{2pt}
}}
{\end{list}}

\newcommand{\norm}[1]{\lVert#1\rVert}

\def\E{\mathbb E}

\newcommand{\R}{\mathbb R}

\newcommand{\kk}{\mathbf{k}}
\newcommand{\Poi}{\mathrm{Poisson}}
\newcommand{\HH}{\mathbb H}
\newcommand{\HC}{\mathcal Q}

\title{Locality Sensitive Hashing for $p$-Exponential Kernels with Applications to Density Estimation}

\author{%
  Barak Gorodissky \\
  Tel Aviv University\\
 \texttt{gorodissky@mail.tau.ac.il} \\
   \And
   Tal Wagner \\
   Tel Aviv University \\
   \texttt{talwag@tauex.tau.ac.il} \\
}

\begin{document}

\maketitle

\begin{abstract}
  A kernel $\kk(x,y)$ is \emph{LSHable} if there exists a locality sensitive hashing scheme $\mathcal H$ such that $\kk(x,y)=\Pr_{h\sim\mathcal H}[h(x)=h(y)]$ for all $x,y$. This notion plays a key role in efficient kernel methods in high dimensions. In this work, we show that the $p$-exponential kernel $\kk(x,y)=\exp(-\norm{x-y}_p)$ is LSHable in bounded regions for all $1<p\leq2$. Previously, this was known only for $p=1$. Our new \emph{mosaic LSH} scheme is based on a Poisson hyperplane process with hyperplanes sampled as $\ell_1$-biased $p$-stable vectors, for which we develop efficient sampling procedures. As applications, our results yield new and efficient density estimation methods based on LSHability for those $p$-exponential kernels.
\end{abstract}

\section{Introduction}\label{sec:intro}
Kernel densities have long been fundamental to machine learning \cite{hofmann2008kernel}. Hashing-based techniques have recently emerged as a powerful approach for kernel in high-dimensional settings where traditional techniques are often rendered infeasible. The central notion is the following:

\begin{definition}[kernel LSH]\label{def:kernellsh}
    Let $\kk:\mathcal X\times\mathcal X\rightarrow[0,1]$ be a map 
    over pairs in $\mathcal X$. A \emph{locality sensitive hashing (LSH) scheme} for $\kk$ is a distribution $\mathcal H$ over hash functions on $\mathcal X$ such that
    \begin{equation} \label{eq:lsh}
    \forall \; x,y\in\mathcal X , \quad\quad \kk(x,y) = \Pr_{h\sim H}[h(x)=h(y)] . 
    \end{equation}
    $\mathcal H$ must be endowed with an algorithm for sampling $h\sim\mathcal{H}$, a finite representation of $h$, and an algorithm to evaluate $x\mapsto h(x)$ given $x\in\mathcal{X}$ and $h$. 
    If such $\mathcal H$ exists, we call $\kk$ an \emph{LSHable kernel}. 
\end{definition}
This notion was studied in Charikar's seminal work on the SimHash algorithm \cite{charikar2002similarity}, and more recently has been central in high-dimensional kernel density estimation (KDE) due to a breakthrough work by Charikar and Siminelakis \cite{charikar2017hashing}. 
Intuitively, $\kk$ encodes the similarity between pairs of elements in $\mathcal X$, and $\mathcal H$ is ``locality sensitive'' because similar elements are more likely to be hashed together. This property is the driving force behind a range of efficient high-dimensional kernel methods, including relative-error KDE \cite{charikar2017hashing,siminelakis2019rehashing,backurs2019space}, streaming KDE \cite{coleman2020sub,liu2023one}, private KDE \cite{coleman2021one,wagner2023fast,wagnerlearning}, kernel ridge regression \cite{zandieh2020scaling}, partition function estimation for importance sampling \cite{spring2017new}, and long-context attention in large language models \cite{joshirace,wolfson2026positional}. 

Despite its centrality and usefulness, there are rather few kernels for which LSHability is known. Classical examples include MinHash \cite{broder1997resemblance} for the Jaccard similarity 
between sets and SimHash \cite{charikar2002similarity} for the angular similarity 
between points on a sphere. 
For kernels over $\R^d$, the Random Binning Features scheme of Rahimi and Recht \cite{rahimi2007random} constitutes an LSH scheme for the 1-exponential kernel, a.k.a~the Laplacian kernel, $\kk(x,y)=\exp(-\norm{x-y}_1)$. Backurs et al.~\cite{backurs2019space} gave a different LSH scheme for the same kernel in a bounded region. This has left as an open problem even the natural case of the 2-exponential kernel,\footnote{Note that the 2-exponential kernel is different from the Gaussian kernel, $\kk(x,y)=\exp(-\norm{x-y}_2^2)$, in which the $\ell_2$ norm in the exponent is squared. We show in Appendix~\ref{sec:nonlsh} that no LSH scheme exists for the Gaussian kernel.} $\kk(x,y)=\exp(-\norm{x-y}_2)$. Applications of LSH-based methods to the 2-exponential and to other kernels have largely gone through various approximations or heuristics.

\paragraph{Our results.}

\begin{table}
  \caption{Summary of LSH results for $\ell_p$ metrics and $p$-exponential kernels.}\label{tbl:main}
  \label{sample-table}
  \centering
  \begin{tabular}{llll}
    \toprule
         &    $p=1$  & $p\in(1,2)$ & $p=2$ \\
    \midrule
    Metric LSH for $\norm{x-y}_p$ &  \cite{indyk1998approximate}  &  \cite{datar2004locality, nguyen2013approximate}   &  \cite{datar2004locality,andoni2008near}  \\
    Kernel LSH for $\exp(-\norm{x-y}_p)$     & \cite{rahimi2007random,backurs2019space} & This work & This work      \\
    \bottomrule
  \end{tabular}
\end{table}

In this work, we design new kernel LSH schemes and prove the LSHability property for all $p$-exponential kernels with $1<p\leq2$ in a bounded region. 
Our main result is the following.

\begin{theorem}\label{thm:main}
    Let $p\in(1,2]$ and $\sigma>0$. There exists a kernel LSH scheme $\mathcal H$ that satisfies \Cref{def:kernellsh} for the $p$-exponential kernel, $\kk(x,y)=\exp(-\norm{x-y}_p/\sigma)$, over points in $\R^d$ with coordinates bounded in $[-1,1]$. 
    The expected time to sample a hash function $h\sim \mathcal H$ and to evaluate $h$ on a point $x$ are each $O(d^2/\sigma)$.
\end{theorem}

The coordinate restriction to $[-1,1]$ 
can be scaled to any interval by inversely scaling the kernel bandwidth $\sigma$. Thus, \Cref{thm:main} provides a kernel LSH scheme for any bounded region. The running time of the LSH scheme grows linearly with the side-length of a bounding hypercube for the region. 

As applications, \Cref{thm:main} leads to new LSH-based algorithms for $p$-exponential kernels. We highlight two examples in \Cref{sec:kde}: relative-error KDE and differentially private KDE.

\subsection{Background and Related Work}\label{sec:related}
\textbf{Metric LSH.} 
LSH was originally introduced 
by Indyk and Motwani \cite{indyk1998approximate} 
for metric spaces in the following form, for the purpose of approximate  nearest neighbor search in high dimensions. 
\begin{definition}[metric LSH \cite{indyk1998approximate}]\label{def:metriclsh}
    Let $(\mathcal X,\mathrm{D})$ be a metric space. A \emph{$(r_1,r_2,p_1,p_2)$-LSH} scheme for $(\mathcal X,\mathrm{D})$ is a distribution $\mathcal H$ over hash functions on $\mathcal X$ such that for each pair $x,y\in\mathcal X$,
    \[
      \mathrm{D}(x,y)\leq r_1 \Rightarrow \Pr_{h\sim\mathcal H}[h(x)=h(y)]\geq p_1 \quad \text{and} \quad \mathrm{D}(x,y)> r_2 \Rightarrow \Pr_{h\sim\mathcal H}[h(x)=h(y)]\leq p_2 . 
    \]
\end{definition}
They designed metric LSH schemes for $\ell_1$ and $\ell_2$ metric spaces. 
Datar et al.~\cite{datar2004locality} designed metric LSH schemes for all $\ell_p$ metrics with $p\in[1,2]$. 
Their approach relies on $p$-stable distributions, which generalize the spherical invariance (i.e., 2-stable) property of the Gaussian distribution \cite{indyk2000stable} (see \Cref{sec:prelim}). 
Andoni and Indyk \cite{andoni2008near} later presented a dramatically improved metric LSH for $\ell_2$, which turned out to be optimal \cite{o2014optimal,motwani2006lower}. 
Their approach was extended to $p\in(1,2)$ by Nguyen \cite{nguyen2013approximate}.

The interest in $\ell_p$ metrics with $p\in(1,2)$ has both mathematical and empirical motivations. Mathematically, these spaces embed isometrically into an infinite-dimensional $L_1$ space, suggesting that algorithmic results known for $\ell_1$ might extend to them. In machine learning, $\ell_p$-based formulations offer a useful interpolation between sparsity-promoting $\ell_1$ formulations and smooth $\ell_2$ formulations, with empirical benefits demonstrated in several applications \cite{kloft2011lp,salzo2018solving,liu2020strongly,varouchakis2025designing,koltsidopoulou2026enhancing}..

\textbf{Kernel LSH.} 
Charikar \cite{charikar2002similarity} focused on kernel LSH as defined in \Cref{def:kernellsh}, which had been previously used for set similarity measures like Broder's MinHash algorithm for the Jaccard similarity \cite{broder1997resemblance}. He showed that kernel LSH is achievable and useful for vectors in $\R^d$, designing the seminal SimHash algorithm. He also showed that if $\kk(x,y)$ is LSHable then $1-\kk(x,y)$ must satisfy the triangle inequality. Further properties of kernel LSH were studied by Chierichetti and Kumar \cite{chierichetti2015lsh}. 

Rahimi and Recht \cite{rahimi2007random} designed the Random Binning Features method for the $1$-exponential (a.k.a. Laplacian) kernel. As observed by Andoni and Indyk in \cite{andoni2009dimension}, this method constitutes a kernel LSH scheme. Backurs et al.~\cite{backurs2019space} gave a different kernel LSH scheme for the same kernel. Their scheme is restricted to points in a bounded region and is faster in certain parameter regimes.

Both of these LSH schemes for the 1-exponential kernel 
share the same underlying technique: imposing a random axis-aligned grid on $\R^d$ and hashing each point to the grid cell that contains it. They differ in how the random grid is selected. 
The randomized grid approach is well-suited for $\ell_1$, albeit less so for other $\ell_p$ spaces. The metric LSH schemes for $p\in[1,2]$ by  \cite{datar2004locality} are also based on randomly shifted grids, composed over projections onto $p$-stable directions. Those projections can be viewed as lossy probabilistic embeddings of $\ell_p$ into $\ell_1$. 
This approach suffices for metric LSH, which implicitly involves an approximation gap (the gap is $c=r_2/r_1$ in \Cref{def:metriclsh}). 
Unfortunately, it does not yield the exact collision identity in \Cref{def:kernellsh} outside the $p=1$ case. 
Therefore, new techniques are needed to achieve kernel LSH for $p\in(1,2]$. Our mosaic LSH scheme departs from randomly shifted grids and relies on Poisson hyperplane processes.. 
\Cref{tbl:main} contains a summary of metric and kernel LSH results for $p\in[1,2]$.

\textbf{Approximate kernel LSH.}
Approximate variants of kernel LSH relax the equality in (\ref{eq:lsh}) to hold only up to an additive \cite{andoni2009dimension,wagner2023fast} or multiplicative \cite{charikar2017hashing,backurs2019space}
  approximation factor. 
Such variants suffice for some LSHability-based methods, with an additional loss depending on the approximation quality.

A prominent example is the $2$-exponential kernel $\exp(-\norm{x-y}_2)$. 
An approximate kernel LSH scheme can be obtained by composing a $\ell_2\rightarrow\ell_1$ 
embedding step with distortion $1\pm\varepsilon$, based on a variant of the Johnson–Lindenstrauss (JL) transform \cite{ailon2009fast,matouvsek2008variants}, with a kernel LSH scheme for $\exp(-\norm{x-y}_1)$. This approach has been used in \cite{charikar2017hashing,backurs2019space,wagner2023fast}
to apply LSHability methods to the 2-exponential kernel.
The embedding distortion parameter $\varepsilon$ controls the discrepancy between the collision probability and the target kernel, and the sampling and evaluation costs depend polynomially on $\varepsilon^{-1}$. 
It therefore offers a trade-off between approximation accuracy and computational cost. 
For $p\in(1,2)$, approximate $\ell_p\rightarrow\ell_1$ embeddings (e.g., \cite{johnson1982embedding}) suggest an analogous route to approximate kernel LSH for $p$-exponential kernels, if they can be made algorithmic. However, such schemes would retain the embedding distortion as a parameter. Our results instead establish exact kernel LSH for all $p\in(1,2]$, removing the embedding distortion as a source of error and computational cost.

\textbf{Random tessellations.} Balog et al.~\cite{balog2016mondrian} showed that the pointwise cell collision probability of the Mondrian partition process~\cite{roy2008mondrian} is equal to the Laplacian kernel. O'Reilly and Tran \cite{o2022stochastic}
extended this to the 2-exponential kernel via an isotropic STIT process. These results establish stochastic processes with the collision identity in \cref{eq:lsh}, but without the sampling and evaluation bounds needed for LSH. 

\subsection{Our Techniques}

\begin{figure}
\captionsetup[subfigure]{labelformat=empty}
 \begin{subfigure}{0.2\textwidth}
     \includegraphics[width=\textwidth]{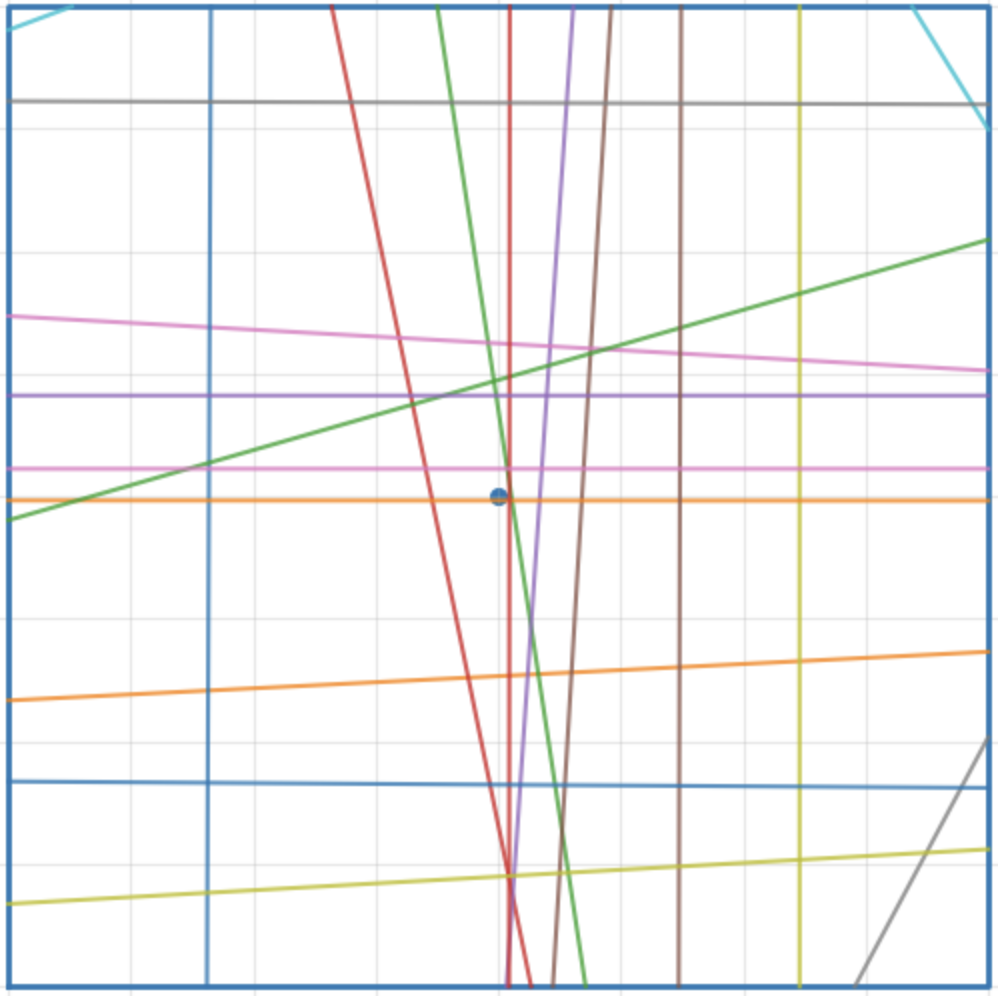}
     \caption{$p=1.25$}
     \label{fig:a}
 \end{subfigure}
 \hfill
 \begin{subfigure}{0.2\textwidth}
     \includegraphics[width=\textwidth]{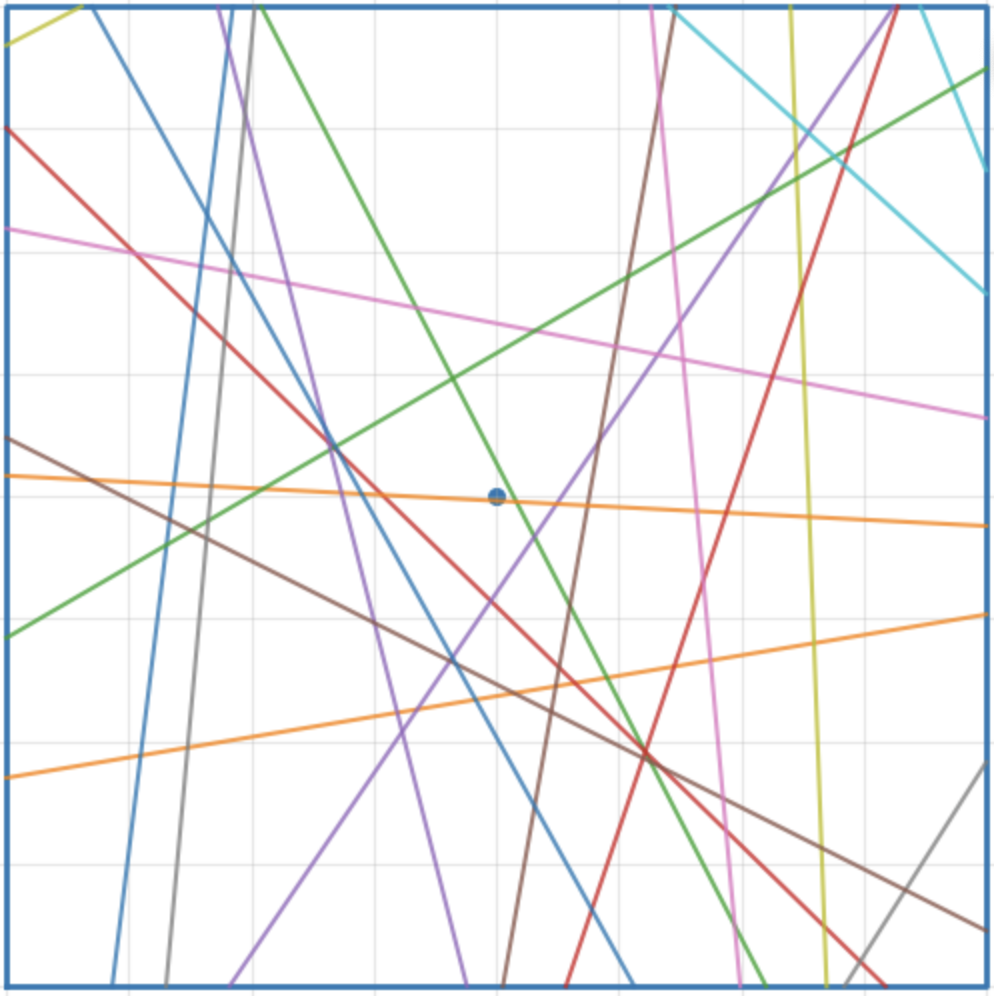}
     \caption{$p=1.5$}
     \label{fig:b}
 \end{subfigure}
  \hfill
 \begin{subfigure}{0.2\textwidth}
     \includegraphics[width=\textwidth]{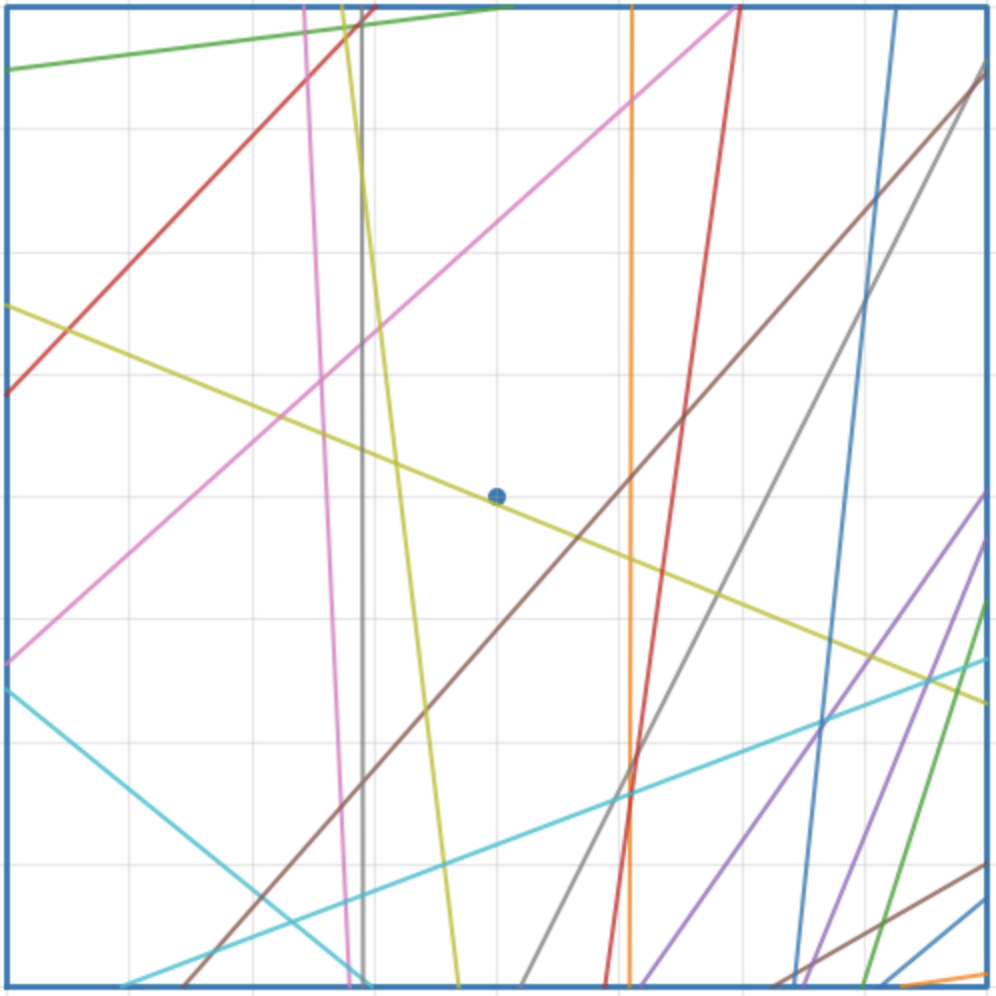}
     \caption{$p=1.75$}
     \label{fig:c}
 \end{subfigure}
  \hfill
 \begin{subfigure}{0.2\textwidth}
     \includegraphics[width=\textwidth]{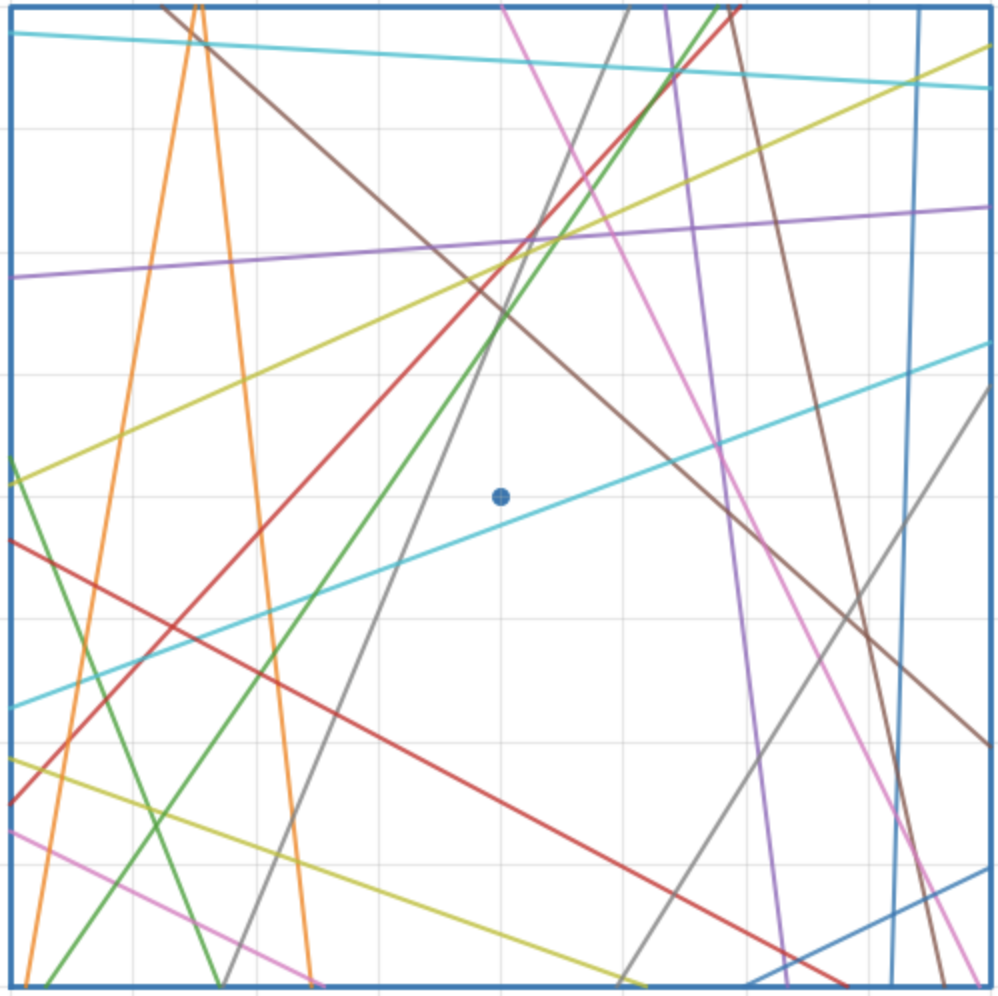}
     \caption{$p=2$}
     \label{fig:d}
 \end{subfigure}
 \caption{Samples from the mosaic LSH in two dimensions for various values of $p$.}\label{fig:vis}
\end{figure}

The first main step in our new approach to kernel LSH is to simulate 
Poisson hyperplane processes (PHPs). These are stochastic processes for generating random hyperplanes in $\R^d$. The sampled hyperplanes induce a partition of $\R^d$ called a \emph{mosaic} \cite{schneider2008stochastic}. The mosaic cells will be our LSH bins. 

The appeal of PHPs in LSH for $p$-exponential kernels is the following. For a compact  subset $A\subset\R^d$ let $\HH(A)$ denote the set of all affine hyperplanes that intersect it. 
The number of samples from $\HH(A)$ in a PHP admits a Poisson distribution with parameter $\lambda_A:=\int_{\HH(A)}d\nu$ where $\nu$ is the intensity measure of the PHP. Two points $x,y\in\R^d$ will form a hash collision if no hyperplanes cross the interval $I_{xy}$ that connects them. This occurs with probability $\Pr[\Poi(\lambda_{I_{xy}})=0]=\exp(-\lambda_{I_{xy}})$. Thus, if we could choose the PHP intensity $\nu$ such that $\lambda_{I_{xy}}=\norm{x-y}_p$, we would be done.

For the $2$-exponential kernel it now seems natural to sample hyperplanes with uniformly random directions---e.g., via i.i.d.~Gaussian entries---and a uniformly random shift from the origin. For $p<2$, the natural generalization is to replace the Gaussian entries with i.i.d.~$p$-stable entries. Indeed, these choices lead to the desired identity $\lambda_{I_{xy}}=\norm{x-y}_p$. Unfortunately, the resulting PHP is not computationally realizable, since $\R^d$ itself is not compact. Simulating the PHP would mean sampling an infinite number of random hyperplanes. Thus, it does not yield a computable LSH scheme.

Our solution to this issue is to restrict the LSH to a compact region $C$ and develop an efficient sampling scheme for the restriction of the PHP to $C$. 
At first blush, it may seem like one could simply sample the hyperplane directions with the same $p$-stable laws and just restrict the random shift to be within the boundaries of $C$. However, this approach fails, as it skews the PHP intensity measure and breaks the key identity $\lambda_{I_{xy}}=\norm{x-y}_p$. Rather, one must identify the actual marginal intensity on $C$ and devise a method to efficiently sample from it. 

Our second main technical step is to show that this is possible for the hypercube $C=\HC:=[-1,1]^d$. We show that the marginal PHP intensity on $\HC$ takes the form of an \emph{$\ell_1$-biased $p$-stable law}. We proceed to show how to sample from it efficiently. For $p=2$, sampling a hyperplane from this law amounts to replacing a single (uniformly random) coordinate in the Gaussian vector with a sample from a two-sided variant of the Rayleigh distribution. For $p<2$, we do not obtain a closed-form sampling procedure, but we are able to present an efficient rejection sampling scheme. 
This enables us to simulate the desired PHP on $\HC$ and obtain the kernel LSH schemes that prove \Cref{thm:main}.

\section{The Mosaic LSH for $p$-Exponential Kernels}\label{sec:lsh}
\subsection{Preliminaries}\label{sec:prelim}

\textbf{Poisson hyperplane processes.}
We introduce notation and review some basic properties of PHPs. We refer to the textbook \cite{schneider2008stochastic} for further background and terminology. 

Let $\HH(A)$ be the set of all affine hyperplanes intersecting $A\subset\R^d$. 
We let
\[ H_{z,t}=\{x\in\R^d:z^Tx=t\} 
\]
denote the hyperplane with normal direction $z\in\R^d$ and shift $t\in\R$ from the origin. It will be convenient to allow $z$ to be unnormalized (thus $z,t$ determine $H_{z,t}$ up to multiplication by a scalar). 

A PHP is a Poisson point process \cite{kingman1992poisson,last2018lectures} over $\HH(A)$. It is defined by an intensity measure $\nu$ over $\R^d\times\R$. For every compact $A\subset\R^d$ with finite non-zero $\nu(\HH(A))$, the number $N(A)$ of hyperplanes intersecting $A$ in the PHP is distributed like $\Poi(\lambda_A)$, where
    \[
      \lambda_A := \int_{\HH(A)}d\nu .
    \]
Conditioned on $N(A)$, the PHP hyperplanes in $\HH(A)$ are distributed like i.i.d.~samples with density
\[\bar{\nu}^A:=\lambda_A^{-1}\cdot \nu|_{\HH(A)} , \]
which is the normalized restriction of $\nu$ to $\HH(A)$. 
This is useful for simulating the PHP.

\begin{lemma}[PHP simulation]\label{lmm:simulation}
    Suppose we have an algorithm $\Pi_A$ for sampling from $\bar\nu^A$ in time $\mathcal{T}_A$. Then, we can simulate the restriction of a PHP with intensity $\nu$ to $A$ in expected time $O(\lambda_A\mathcal{T}_A)$ by sampling $N(A)\sim\Poi(\lambda_A)$ and then running $\Pi_A$ for $N(A)$ times. 
\end{lemma}

\textbf{$p$-Stable distributions.}
We denote by $\mathcal D_p$ the distribution with characteristic function $e^{-|t|^p}$. Note that $\mathcal{D}_2$ is the Gaussian distribution $N(0,2)$. 
More generally, $\mathcal D_p$ is known to exist for $0<p\leq2$. It is well-known how to sample from $\mathcal D_p$ in time $O(1)$ via the Chambers-Mallows-Stuck (CMS) transform \cite{chambers1976method}, which generalizes the classical Box-Muller transform for Gaussians \cite{scott2011box}.

Let $\mathcal{D}_p^{\otimes d}$ denote the distribution over $\R^d$ with i.i.d. entries from $\mathcal{D}_p$. 
If $Z$ is sampled from $\mathcal{D}_p^{\otimes d}$, then $Z^Tu$ is distributed like $\norm{u}_p\mathcal{D}_p$ for every fixed $u\in\R^d$. 
This is the $p$-stable property of $\mathcal{D}_p$.
 
Let 
$m_p:=\E_{X\sim\mathcal{D}_p}|X|$
denote the first absolute $p$-stable moment. It is known that $m_p$ is finite for $p\in(1,2]$ and admits the closed form $m_p=2\pi^{-1}\Gamma(1-1/p)$ \cite[Theorem 5.1]{janson2011stable}.

\textbf{Kernel bandwidth.}
The $p$-exponential kernel is $\kk(x,y)=\exp(-\norm{x-y}_p/\sigma)$ with bandwidth parameter $\sigma>0$. 
In \Cref{sec:intro} we let $\sigma=1$ for simplicity. In this section, it will be more convenient to assume w.l.o.g.~that our bounded region is the hypercube $\HC=[-1,1]^d$ and keep $\sigma$ explicit. This maintains full generality since \Cref{thm:main} can be applied to any bounded region by scaling $\sigma$. 

\subsection{LSH via PHPs}\label{sec:lshviaphp}
A collection of hyperplanes $\{H_{z^{(j)},t^{(j)}}\}_{j=1}^m$ induces a polygonal space partition called a mosaic. We use the cells of the mosaic to induce a hash function $h$ over $\R^d$. Each point $x\in\R^d$ can be hashed to its cell in time $O(md)$ by computing the length-$m$ bitstring 
$h(x)_j=\mathbf1\{(z^{(j)})^Tx-t^{(j)}\geq0\}$, which encodes whether $x$ is above or below each hyperplane in the collection. See \Cref{fig:vis} for illustration.

Our LSH schemes will use hyperplanes sampled by simulating appropriate PHPs. 
For every pair $x,y\in\R^d$, the collision event $h(x)=h(y)$ occurs if the PHP avoids the set $\HH(I_{xy})$, where $I_{xy}\subset\R^d$ denotes the interval between $x$ and $y$. Since $N(I_{xy})$ is distributed like $\Poi(\lambda_{I_{xy}})$, the avoidance probability is $\exp(-\lambda_{I_{xy}})$. 
To obtain \Cref{thm:main} we need a PHP such that $\lambda_{I_{xy}}=\norm{x-y}_p/\sigma$.

We now define the intensity measure $\nu$ of our PHP over hyperplanes $H_{Z,T}$. It is the product measure where $Z$ has measure
$\frac{1}{\sigma m_p}\mathcal{D}_p^{\otimes d}$ and $T$ has the uniform measure over $\R$. 
Under this choice of $\nu$, we compute $\lambda_{I_{xy}}$ for a given pair $x,y\in\R^d$ as follows (expectations henceforth are over $Z\sim\mathcal{D}_p^{\otimes d}$):
\begin{align*}
  \lambda_{I_{xy}} &= \int_{\HH(I_{xy})}d\nu \\
  &= \frac{1}{\sigma m_p}\E\left[\int_{\min\{Z^Tx,Z^Ty\}}^{\max\{Z^Tx,Z^Ty\}}dt\right] \\
  &= \frac{1}{\sigma m_p}\E\left|Z^Tx-Z^Ty\right| \\
  &= \frac{1}{\sigma m_p}\E|Z^T(x-y)| .
\end{align*}
Recall that the $p$-stable property implies that $Z^T(x-y)$ is distributed like $\norm{x-y}_p\mathcal{D}_p$, and hence
\[ \E[|Z^T(x-y)|]=\norm{x-y}_p\E_{X\sim \mathcal{D}_p}|X|=\norm{x-y}_pm_p . \]
Therefore, $\lambda_{I_{xy}}=\norm{x-y}_p/\sigma$ as desired. 

\subsection{Restriction to the Hypercube}\label{sec:hypercube}
The PHP described above cannot be simulated in finite time over $\R^d$ since $\nu(\HH(\R^d))$ is not finite. 
We now show how to simulate it over the hypercube $\HC=[-1,1]^d$ using \Cref{lmm:simulation}. That lemma requires \emph{(i)} calculating $\lambda_{\HC}$, and \emph{(ii)} developing a sampling algorithm to draw a sample from $\bar\nu^\HC$. 

We start with $\lambda_{\HC}$. Observe that for every fixed $z\in\R^d$, the inner product $z^Tx$ is maximized over $x\in\HC$ where all coordinates of $x$ have absolute value 1 and their signs match the respective coordinates of $z$. The same holds for minimizing $z^Tx$ but with the signs of $x$ flipped. Therefore,
\begin{equation}\label{eq:l1}
  l_z(\HC) := \min_{x\in\HC}z^Tx = -\sum_{i=1}^d|z_i|=-\norm{z}_1 \quad \text{and} \quad u_z(\HC) :=\max_{x\in\HC}z^Tx =\sum_{i=1}^d|z_i|=\norm{z}_1 .
\end{equation}
Therefore: 
\begin{equation}\label{eq:lambdaq}
    \lambda_{\HC} = \int_{\HH(\HC)}d\nu  
    = \frac{1}{\sigma m_p}\E\left[\int_{l_Z(\HC)}^{u_Z(\HC)}dt\right] 
    = \frac{1}{\sigma m_p}\E\left[2\sum_{i=1}^d|Z_i|\right]
    = \frac{2}{\sigma m_p}\sum_{i=1}^d\E|Z_i|
 = \frac{2d}{\sigma},
\end{equation}
%
where we have used that each $Z_i$ has distribution $\mathcal{D}_p$, and is independent of the other coordinates $\{Z_j:j\neq i\}$ of $Z$, and therefore $\E|Z_i|=m_p$. Plugging \Cref{eq:lambdaq} in the simulation lemma \Cref{lmm:simulation}, we find that simulating the PHP from \Cref{sec:lshviaphp} would take expected time $O(\mathcal{T}_{\HC}\cdot d/\sigma)$ if we had an algorithm for sampling from $\bar\nu^\HC$ in time $\mathcal{T}_{\HC}$. 

Next, we show how to sample from $\bar\nu^\HC$, which we recall is the probability density proportional to the restriction of $\nu$ to $\HC$. We denote by $\bar\nu^{\HC}_{Z}(z)$ the $Z$-marginal density of $\bar\nu^\HC$ and by $\bar\nu^{\HC}_{T|Z}(t|z)$ the density of $T$ in $\HC$ conditioned on $Z$. Thus we have $\bar\nu^{\HC}(t,z) = \bar\nu^{\HC}_{T|Z}(t|z)\cdot \bar\nu^{\HC}_{Z}(z)$. Observe that $\bar\nu^{\HC}_{T|Z}(t|z)$ is simply the uniform measure over shifts $t\in\R$ such that $H_{z,t}$ intersects $\HC$. Therefore, we can sample from $\bar\nu^\HC$ by first sampling $z$ from $\bar\nu^{\HC}_Z$ and then sampling $t$ uniformly from $[l_Z(\HC),u_Z(\HC)]$.

It thus remains to show how to draw a sample from $\bar\nu^{\HC}_Z$. For each fixed $z\in\R^d$ we have
\begin{equation}\label{eq:vqz}
    \bar\nu_Z^{\HC}(z) = \int_{-\infty}^\infty d\bar\nu^{\HC}(z,t) 
    = \frac{1}{\lambda_{\HC}}\cdot\frac{1}{\sigma m_p}\cdot \mathcal{D}_p^{\otimes d}(z)\int_{l_z(\HC)}^{u_z(\HC)}dt 
    = \frac{\sigma}{2d}\cdot\frac{1}{\sigma m_p}\cdot \mathcal{D}_p^{\otimes d}(z) \cdot 2\norm{z}_1 ,
\end{equation}
where we have plugged (\ref{eq:l1}) and (\ref{eq:lambdaq}) for the final equality. We see that the probability density $\bar\nu_Z^{\HC}(z)$ 
equals $\frac{1}{dm_p}\norm{z}_1\mathcal{D}_p^{\otimes d}$. This
is proportional to $\norm{z}_1\mathcal{D}_p^{\otimes d}$, i.e., it is an \emph{$\ell_1$-biased $p$-stable distribution}. 
This is an instance of a size-biased distribution, a familiar topic in probability theory. We refer to \cite{arratia2019size} for background on size-biased distributions. 

In summary, we proved the following lemma:
\begin{lemma}\label{lmm:reductiontosampling}
    Suppose we have an algorithm for sampling $z\in\R^d$ with density $\frac{1}{dm_p}\norm{z}_1\mathcal{D}_p^{\otimes d}$ in time $\mathcal{T}$. Then we can simulate in expected time $O(\mathcal Td/\sigma)$ a PHP over $\HC$ where the avoidance probability of $\HH(I_{xy})$ is $\exp(-\norm{x-y}_p/\sigma)$ for every $x,y\in\HC$. 
\end{lemma}
As a consequence, we get:
\begin{corollary}\label{cor:lsh}
Suppose we have an algorithm for sampling $z\in\R^d$ with density $\frac{1}{dm_p}\norm{z}_1\mathcal{D}_p^{\otimes d}$ in time $\mathcal{T}$. Then, we have an LSH scheme over $\HC$ with sampling time $O(\mathcal Td/\sigma)$, evaluation time $O(d^2/\sigma)$, and collision probability $\exp(-\norm{x-y}_p/\sigma)$.  
\end{corollary}
\begin{proof}
    \Cref{sec:lshviaphp} shows that the PHP from \Cref{lmm:reductiontosampling} induces an LSH scheme with the desired collision probability. 
    Sampling a hash function from the LSH scheme means simulating the PHP. 
    Using \Cref{lmm:simulation}, the simulation time is $O(N\mathcal{T})$ where $N$ is the number of sampled hyperplanes. Since $N$ is distributed like $\Poi(2d/\sigma)$, the expected LSH sampling time is $O(\mathcal{T}d/\sigma)$.
    The time to evaluate a sampled LSH map $h$ on a point $x\in\R^d$ is the time it takes to compute the inner product of $x$ with the normal of each hyperplane. Each inner product takes time $O(d)$, and since the number of hyperplanes $N$ is in expectation $2d/\sigma$, the total expected time to evaluate $h$ on $x$ is $O(d^2/\sigma)$. 
\end{proof}

\subsection{Sampling $\ell_1$-Biased $p$-Stable Vectors}\label{sec:l1biased_pstable_sampling}
Our goal in this section is to develop a method to sample a vector $z\in\R^d$ with density proportional to $\norm{z}_1\mathcal{D}_p(z)$. By \Cref{cor:lsh}, this will complete the sampling algorithm for $\bar\nu^\HC$, which will allow us to simulate the PHP on $\HC$ with \Cref{lmm:simulation}, thus proving \Cref{thm:main}. 

Formally, this section proves the following lemma.

\begin{lemma}\label{lmm:samplesizebiased}
 For every $p\in(1,2]$ there is an algorithm for sampling $z\in\R^d$ with density $\frac{1}{dm_p}\norm{z}_1\mathcal{D}_p^{\otimes d}$ in time $O(d)$ plus expected time $O(1)$. For $p=2$,
 the time is just $O(d)$.     
\end{lemma}

To start, we observe that the additivity of the $\ell_1$ norm is useful for reducing the problem of sampling according to the $\ell_1$-biased law from sampling a vector to sampling a single coordinate. 

\begin{lemma}\label{lmm:sizebiasedp}
    We can sample $z\in\R^d$ with density $\frac{1}{dm_p}\norm{z}_1\mathcal{D}_p(z)$ by the following process:
    \begin{CompactEnumerate}
                \item Sample $I$ uniformly from $\{1,\ldots,d\}$.
        \item For all $j\neq I$, sample $Z_j$ from $\mathcal D_p$ (scalar $p$-stable).
        \item Sample $Z_I$ with density  $m_p^{-1}|z|\mathcal D_p(z)$ (scalar magnitude-biased $p$-stable).
        \item Return the sampled vector $(Z_1,\ldots,Z_d)$. 
    \end{CompactEnumerate}
\end{lemma}
\begin{proof}
    We calculate the density $p_Z(z)$ of $z\in\R^d$ sampled with the above sampling process:
    \begin{align*}
        p_Z(z) &= \sum_{i=1}^d\Pr[I=i]\cdot p_{Z|I}(z|i) & \text{total probability formula} \\
        &= \sum_{i=1}^d\frac1d\cdot \prod_{j=1}^dp_{Z_j|I}(z|i) & \text{$I$ is uniformly random} \\
        &= \sum_{i=1}^d\frac1d\left(\frac{|z_i|}{m_p}\mathcal D_p(z_i) \cdot \prod_{j\neq i}\mathcal D_p(z_j)\right) & \text{$Z_i\sim\frac{|z_i|}{m_p}\mathcal{D}_p$, and $Z_j\sim\mathcal{D}_p$ for $j\neq i$} \\
        &= \frac{1}{dm_p} \mathcal{D}_p^{\otimes d}(z)\sum_{i=1}^d|z_i| & \text{rearranging} \\
        &=\frac{1}{dm_p}\norm{z}_1\mathcal{D}_p^{\otimes d}(z) & \text{rearranging}
    \end{align*}
    as needed.
\end{proof}
At this point, to complete the sampling algorithm for \Cref{lmm:samplesizebiased}, it remains to show how to sample a scalar $z$ from the magnitude-biased $p$-stable density $m_p^{-1}|z|\mathcal{D}_p(z)$ (step 3 in \Cref{lmm:sizebiasedp}). 

\paragraph{The $p=2$ case.} In this case, we can get a closed-form sampling solution. 
Recall that the 2-stable distribution is $N(0,2)$ and its density is $\frac{1}{2\sqrt{\pi}}e^{-z^2/4}$.  
The corresponding magnitude-biased density is $\frac14|z|e^{-z^2/4}$. It is closely related to the Rayleigh distribution with scale $\sqrt2$, which is supported on $[0,\infty)$ with density $\tfrac12ze^{-z^2/4}$. 
The cumulative density function (CDF) of $\mathrm{Rayleigh}(\sqrt2)$ is $1-e^{-z^2/4}$, and therefore we can sample $R\sim \mathrm{Rayleigh}(\sqrt2)$ by drawing $U\sim(0,1)$ uniformly at random and returning $R=2\sqrt{-\ln(1-U)}$. 
Since we need density $\frac14|z|e^{-z^2/4}$ over $\R$ rather than $\frac12ze^{-z^2/4}$ over $[0,\infty)$, we can sample a random sign $S\sim\mathrm{Uniform}\{\pm1\}$ and return the ``two-sided Rayleigh'' sample $SR$. This is a sample from the magnitude-biased 2-stable density, $m_2^{-1}|z|\mathcal{D}_2(z)$. 

\paragraph{The $1<p<2$ case.}
For this case, we recall the CMS transform \cite{chambers1976method}, which gives a method for sampling $p$-stables by expressing $\mathcal{D}_p$ as the push-forward measure of closed-form distributions. 
Let
\[
  F_p(\theta) := \frac{\sin(p\theta)}{(\cos \theta)^{1/p}\cdot\cos((p-1)\theta)^{(p-1)/p}}
  \quad \text{and} \quad
  T_p(\theta,w) := F_p(\theta)\cdot w^{(p-1)/p} .
\]
\begin{lemma}[CMS transform for $p$-stables \cite{chambers1976method}]\label{lmm:cms}
    Let $\mu$ be the product distribution $\mathrm{Uniform}(-\tfrac{\pi}2,\tfrac{\pi}2)\times\mathrm{Exp}(1)$ and let $(\Theta,W)\sim\mu$.  
    Then $Z=T_p(\Theta,W)$ is a sample from $\mathcal{D}_p$. 
\end{lemma}
By \Cref{lmm:cms} and a standard change of variables for push-forward measures, we get the following:
\begin{corollary}\label{cor:samplebiasedp}
    Let $\bar\mu(\bar \theta,\bar w)$ be the probability density proportional to $|T_p(\bar \theta,\bar w)|\cdot \mu(\bar \theta,\bar w)$. If $(\bar \Theta,\bar W)$ are sampled from $\bar\mu$, then $\bar Z=T_p(\bar \Theta,\bar W)$ has density proportional to $|z|\cdot\mathcal{D}_p(z)$.
\end{corollary}
Both $\mu$ and $|T_p(\cdot,\cdot)|$ decompose over their two coordinates, 
hence $\bar\mu\propto|T_p|\cdot\mu$ is a product measure on $\bar\Theta$ and $\bar W$, so 
we can sample $\bar\Theta$ and $\bar W$ independently. 
Since the $W$-marginal of $\mu$ is $\mathrm{Exp}(1)$, whose density is $e^{-w}$, the $\bar W$-marginal of $\bar\mu$ has density $w^{(p-1)/p}e^{-w}$, which is just $\mathrm{Gamma}(1+\tfrac{p-1}{p},1)$. 

For $\bar\Theta$, while we do not have a closed form for the density, we can derive an efficient rejection sampling scheme. 
Recall that in rejection sampling we have a target density $f_{\bar\Theta}(\theta)$, a proposal density $g_V(\theta)$, and an envelope bound $m>0$ that satisfies $f_{\bar\Theta}(\theta)\leq m \cdot g_V(\theta)$ for all $\theta$ in the support of $\bar\Theta$. We iteratively sample $V$ with density $g_V$, accept the sample with probability $f_{\bar\Theta}(V)/(m\cdot g_V(V))$, and otherwise we reject the sample and repeat. An accepted sample is distributed with density $f_{\bar\Theta}(\theta)$, and the expected number of samples until a sample is accepted is at most $m$. 

Recall that $Y\sim\mathrm{Beta}(\beta,1)$ is supported on $(0,1)$ with density $\beta y^{\beta-1}$. Our rejection sampling scheme for $\bar\Theta$ is specified in the following lemma.  The proof is in Appendix \ref{app:rejection}.

\begin{lemma}\label{lmm:rejection}
Let $V$ be a random variable sampled as $V=S\cdot\mathrm{arccos(Y)}$, where $Y\sim\mathrm{Beta}(1-1/p,1)$ and $S\in\{1,-1\}$ is a uniformly random sign. 
Then, rejection sampling with $V$ as the proposal for $\bar\Theta$ has envelope bound $m=p\leq2$. Therefore, a sample of $\bar\Theta$ is returned in expected time $p=O(1)$.  
\end{lemma}

In summary, by \Cref{cor:samplebiasedp}, we can sample a scalar magnitude-biased $p$-stable $z\in\R$, i.e., with density proportional to $|z|\mathcal{D}_p(z)$, as $z=T_p(\bar\Theta,\bar{W})$ where $\bar{W}\sim\mathrm{Gamma}(1+\tfrac{p-1}{p},1)$ and $\bar\Theta$ is sampled with the rejection sampling scheme from \Cref{lmm:rejection}.
In conjunction with \Cref{lmm:sizebiasedp}, the proof of \Cref{lmm:samplesizebiased} is now complete:
\begin{proof}[Proof of \Cref{lmm:samplesizebiased}]
Sample $I\in\{1,\ldots,d\}$ uniformly at random. For $p=2$, Sample $R\sim\mathrm{Rayleigh}(\sqrt 2)$ and $S\sim\{1,-1\}$. Set $z_I=SR$, and sample the other coordinates of $z$ iid from the normal distribution $N(0,2)$. The sampling time per coordinate is $O(1)$. 
For $p\in(1,2)$, 
sample $\bar{W}\sim\mathrm{Gamma}(1+\tfrac{p-1}{p},1)$ and $\bar\Theta$ with the rejection sampling scheme from \Cref{lmm:rejection}, in expected time $O(1)$. Set $z_I=T_p(\bar\Theta,\bar{W})$, and sample the other coordinates of $z$ iid from $\mathcal{D}_p$ with the CMS transform (\Cref{lmm:cms}) in time $O(1)$ per coordinate. 
In both cases, per above, $z_I$ has the magnitude-biased $p$-stable density, and the other coordinates are usual $p$-stables. Therefore, by \Cref{lmm:sizebiasedp}, the returned vector $z$ has density $\frac{1}{dm_p}\norm{z}_1\mathcal{D}^{\otimes d}_p$ as required by \Cref{lmm:samplesizebiased}.
\end{proof}

\subsection{The Final Mosaic LSH Scheme}
To conclude,  
we collect the steps to sample a hash function from the mosaic LSH for the $p$-exponential kernel (with $p\in(1,2]$):
\begin{CompactEnumerate}
    \item[1.] Sample $N\sim\Poi(2d/\sigma)$. 
    \item[2.] For $j=1,\ldots,N$:
    \begin{CompactEnumerate}
        \item[2.1.] Sample $Z^{(j)}\in\R^d$ with the process in \Cref{lmm:sizebiasedp}. 
        \begin{CompactItemize}
            \item For step 2 in \Cref{lmm:sizebiasedp} (sampling a $p$-stable):
            \begin{CompactItemize}
                \item[-] If $p=2$, sample from $N(0,2)$. 
                \item[-] If $p\in(1,2)$, sample $\Theta\sim[-\tfrac{\pi}2,\tfrac{\pi}2]$ and $W\sim\mathrm{Exp}(1)$ and return $T_p(\Theta,W)$ (\Cref{lmm:cms}). 
            \end{CompactItemize} 
            \item For step 3 in \Cref{lmm:sizebiasedp} (sampling a magnitude-biased $p$-stable): 
            \begin{CompactItemize}
                \item[-] If $p=2$, sample $R\sim\mathrm{Rayleigh}(\sqrt2)$ and a uniformly random sign $S$ and return $SR$. 
                \item[-] If $p\in(1,2)$, sample $\bar \Theta$ using the rejection sampling scheme in \Cref{lmm:rejection} and $\bar W\sim\mathrm{Gamma}(1+\tfrac{p-1}{p},1)$ and return $T_p(\bar \Theta, \bar W)$ (\Cref{cor:samplebiasedp}). 
            \end{CompactItemize} 
        \end{CompactItemize}
        \item[2.2.] Sample $T^{(j)}$ uniformly in $[-\norm{Z^{(j)}}_1,\norm{Z^{(j)}}_1]$
    \end{CompactEnumerate}
    \item[3.] The hyperplanes $\{H_{Z^{(j)},T^{(j)}}:j=1,\ldots,N\}$ induce the mosaic LSH hash function.
\end{CompactEnumerate}

\textbf{Proof of \Cref{thm:main}.} 
By \Cref{lmm:samplesizebiased}, the sampling scheme in step 2.1 above returns an $\ell_1$-biased $p$-stable vector in at most expected time $O(d)$. Plugging this in \Cref{cor:lsh} yields \Cref{thm:main}. \qed
 
\section{Applications to Kernel Density Estimation}\label{sec:kde}

The KDE function of a kernel $\kk$ and a dataset $X$ is $KDE_X(y)=\frac1{|X|}\sum_{x\in X}\kk(x,y)$. 
Many efficient methods for KDE in high dimensions are based on kernel LSH schemes. 
\Cref{thm:main} makes these methods available to $p$-exponential kernels with $p\in(1,2]$. We highlight two applications. 

\textbf{Differentially private KDE.}
Differential privacy (DP) is widely considered the gold standard for privacy-preserving computation in machine learning. 
It formalizes aggregate computations that leak almost no information about any particular point in the dataset. 
We refer to \cite{dwork2006calibrating,dwork2014algorithmic} and to Appendix \ref{app:dpkde} for definitions and background.

Coleman and Shrivastava \cite{coleman2021one} and Wagner et al.~\cite{wagner2023fast} proved that kernel LSH schemes imply efficient DP KDE mechanisms in the function release model. In this model, a privacy-preserving representation of the KDE function can be released and queried infinitely many times without further privacy leakage. Their approach enabled DP KDE mechanisms for the 1-exponential kernel through its LSHability, and, by extension, also for the 2-exponential kernel via an efficient $\ell_2\rightarrow\ell_1$ embedding. Backurs et al.~\cite{backurs2024efficiently} gave a different DP KDE mechanism for the $2$-exponential kernel via a more complex reduction to the Gaussian kernel, which in turn was handled in \cite{wagner2023fast} with a different technique. 

\Cref{thm:main} produces the first explicit DP KDE mechanisms for $p$-exponential kernels with $p\in(1,2)$, by plugging it into the LSHability framework of \cite{wagner2023fast} for DP KDE. See Appendix \ref{app:dpkde} for details.

\begin{theorem}\label{thm:dpkde}
Let $\epsilon,\sigma,d>0$, let $\alpha,\eta\in(0,1)$, and let $p\in(1,2)$. 
For every set $X\subset\HC\subset\R^d$ that satisfies $|X|\geq\Omega(\log(1/\eta)/(\alpha^2\epsilon))$, 
the kernel $\kk(x,y)=\exp(-\norm{x-y}_p/\sigma)$ in $\HC$ admits an $\epsilon$-DP KDE mechanism in the function release model for the KDE function of $X$, with additive error $\alpha$, failure probability $\eta$, and expected query time $O(d^2\log(1/\eta)/(\alpha^2\sigma))$. 
\end{theorem}

\textbf{KDE with relative error via hashing-based estimators (HBE).}
Charikar and Siminelakis \cite{charikar2017hashing} studied the following form of KDE with relative error. Given $\varepsilon,\mu\in(0,1)$, the goal is to build a data structure that for each $y$ such that $KDE_X(y)\geq\mu$, it returns an estimate of $KDE_X(y)$ within $(1\pm\varepsilon)KDE_X(y)$ with high probability. This is useful for small KDE values $\mu\leq KDE_X(y)\ll \varepsilon$ where an additive error of $\pm\epsilon$ is considered too large. 
By standard concentration, sampling $O(1/\varepsilon^2\mu)$ uniformly random points from $X$ and evaluating the KDE only on them yields a data structure for KDE with this relative error guarantee. 
Remarkably, \cite{charikar2017hashing} showed this can be achieved with a much smaller number of samples, $1/\varepsilon^2\mu^\alpha$ with $\alpha<1$. Their HBE method achieves $\alpha=0.5$ for LSHable kernels (see also Algorithm 1 in \cite{backurs2019space}). Plugging \Cref{thm:main} into the HBE, we obtain this result for $p$-exponential kernels with $1<p<2$ (the cases $p=1,2$ were already possible with HBE). 

\begin{theorem}\label{thm:hbe}
    Let $p\in(1,2)$. The HBE method \cite{charikar2017hashing,backurs2019space} with the kernel LSH schemes from \Cref{thm:main} yields a data structure for KDE with relative error which uses $O(1/\varepsilon^2\sqrt\mu)$ samples ($\alpha=0.5$). 
\end{theorem}

Later work by Charikar et al.~\cite{charikar2020kernel} showed that KDE with relative error can also be obtained from metric LSH rather than kernel LSH, with even better asymptotic bounds. For the $2$-exponential kernel, \cite{charikar2020kernel} obtained $\alpha=0.1$. They did not explicitly handle $p<2$. However, plugging the classical metric LSH schemes for $\ell_p$ metrics \cite{datar2004locality} into their framework yields data structures with $\alpha=0.25$ for $p$-exponential kernels with $p\in[1,2)$. As a result, \Cref{thm:main} does not lead to state-of-the-art asymptotic bounds for relative-error KDE. Nevertheless, the downside of \cite{charikar2020kernel} is that the method involves substantial combinatorial complexity and seems hard to implement. In contrast, the HBE method has been implemented and shown to be practical in several works \cite{siminelakis2019rehashing,backurs2019space,backurs2021faster}. Therefore, \Cref{thm:main} yields the first \emph{practical and provable} method for KDE with a relative-error guarantee for $p\in(1,2)$. 
The next section includes empirical evaluations for \Cref{thm:dpkde,thm:hbe}.

\section{Experimental Evaluation}\label{sec:experiments}
In this section, we include experimental evaluations to validate our theoretical results. 
For the values of $p$ we use $p\in\{1.25,1.5,1.75,2\}$. 
We run experiments on two standard benchmark KDE datasets: Covertype \cite{covertype_31} and U.S.~Census \cite{us_census_data_1990_116}. 
Both have been used in prior work on LSH-based methods for KDE \cite{siminelakis2019rehashing,backurs2019space,backurs2021faster,wagner2023fast}.  
Bandwidth values and other implementation details are specified in Appendix~\ref{app:additional_results} for reproducibility. 
Due to space limits, in this section we present the results for Covertype. The results for Census are qualitatively similar and are in Appendix~\ref{app:additional_results}. Our code is available online at \texttt{\href{https://github.com/gorobarak/pExponentialLSH}{https://github.com/gorobarak/pExponentialLSH}}.

\textbf{(i) Kernel estimation.}
Our first experiment validates eq.~(\ref{eq:lsh}) directly by comparing the empirical collision probability of the mosaic LSH to the true kernel value. The results in \Cref{fig:direct_kernel_est_covertype} show that the error decays quickly as the number of LSH functions sampled increases.

\textbf{(ii) Differentially private KDE.} 
We implement the LSH-based mechanism for DP-KDE from \cite{wagner2023fast}. 
The mechanism uses a kernel LSH scheme to generate random features for the data points and then adds privacy-preserving noise to their average. We implement it with the mosaic LSH schemes from \Cref{thm:main} and with $100$ random features.
For $p=2$, we include as a baseline the same mechanism with the approximate LSH for the 2-exponential kernel that arises by applying an $\ell_2\rightarrow\ell_1$ JL embedding and then using the 1-exponential kernel LSH scheme from \cite{backurs2019space}. For the other values of $p$, our results yield the first explicit DP KDE mechanisms, so there is no prior baseline for comparison. 

The results in \Cref{fig:lsq_eps_est_covertype} show that for all values of $p$, the KDE error decays as the privacy budget $\varepsilon$ increases. For $p=2$, our exact kernel LSH and the JL-based approximate LSH perform very similarly (note that DP KDE is inherently approximate due to its built-in privacy-preserving noise). 

\begin{figure}
    \centering
    \includegraphics[width=0.8\textwidth]{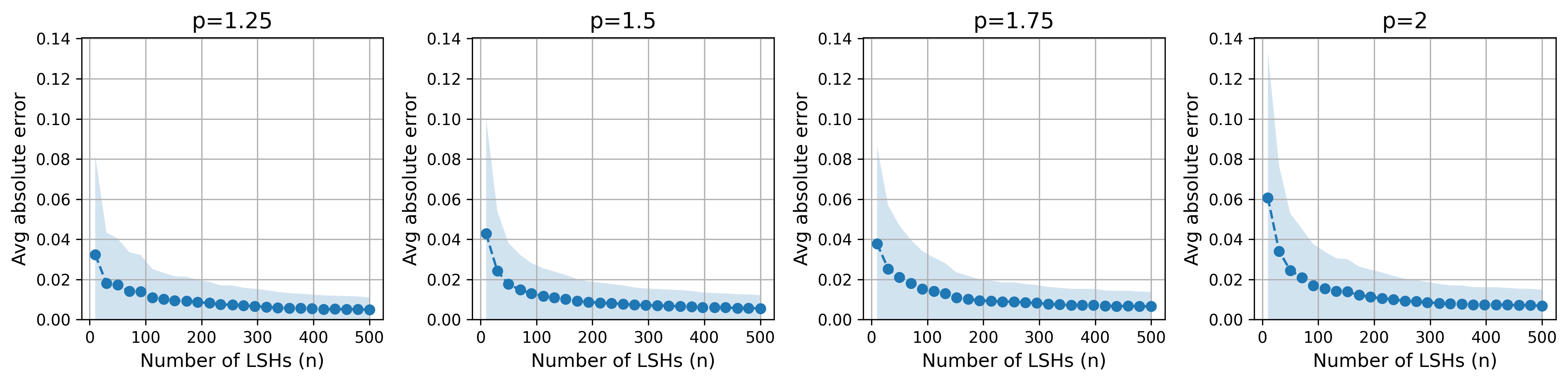}
    \caption{$p$-Exponential kernel estimation by mosaic LSH on the Covertype dataset}
\label{fig:direct_kernel_est_covertype}
\vskip 1em
    \includegraphics[width=0.8\textwidth]{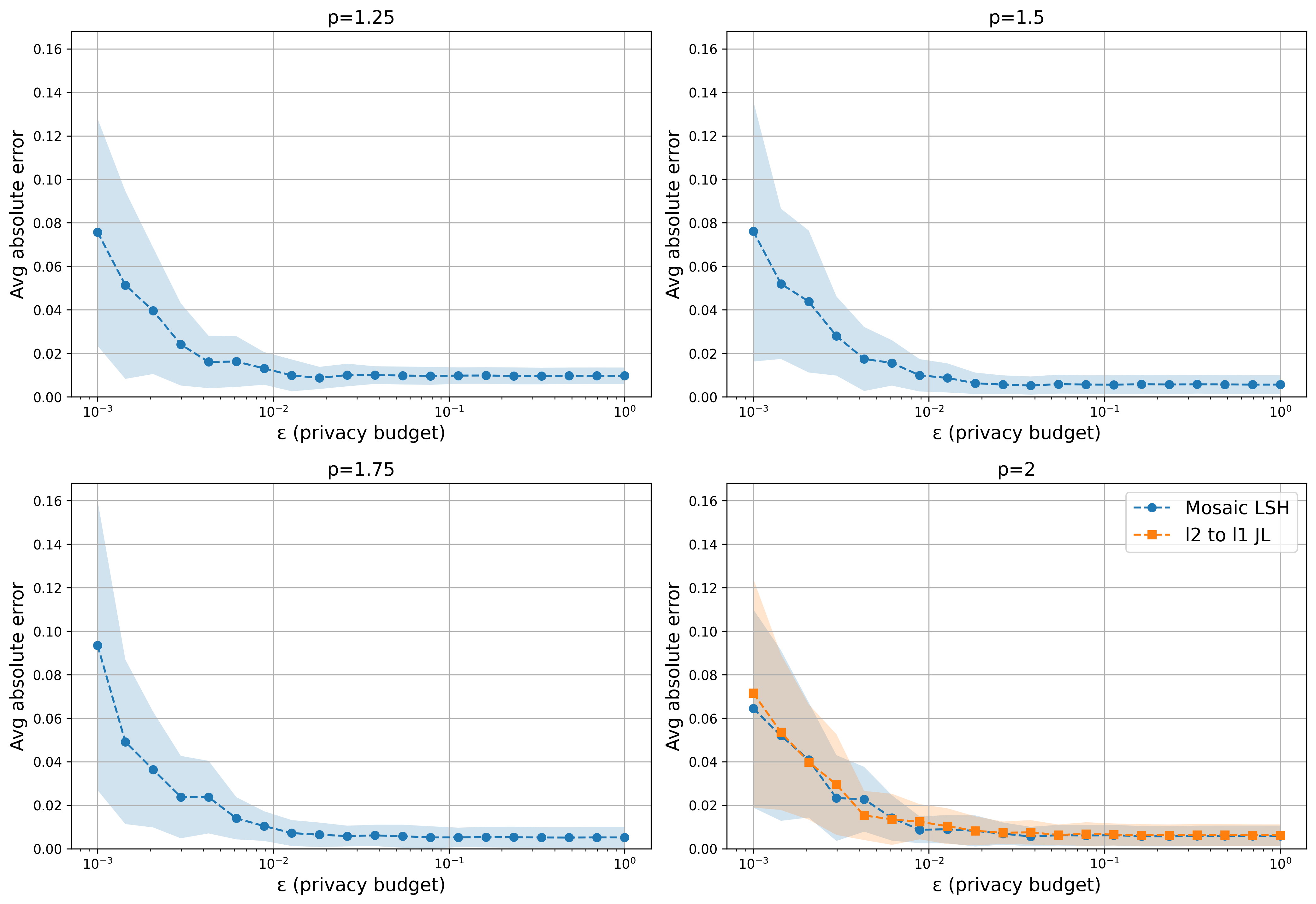}
    \caption{Differentially private KDE results for Covertype}

    \label{fig:lsq_eps_est_covertype}
    \vskip 1em
    \includegraphics[width=0.8\textwidth]{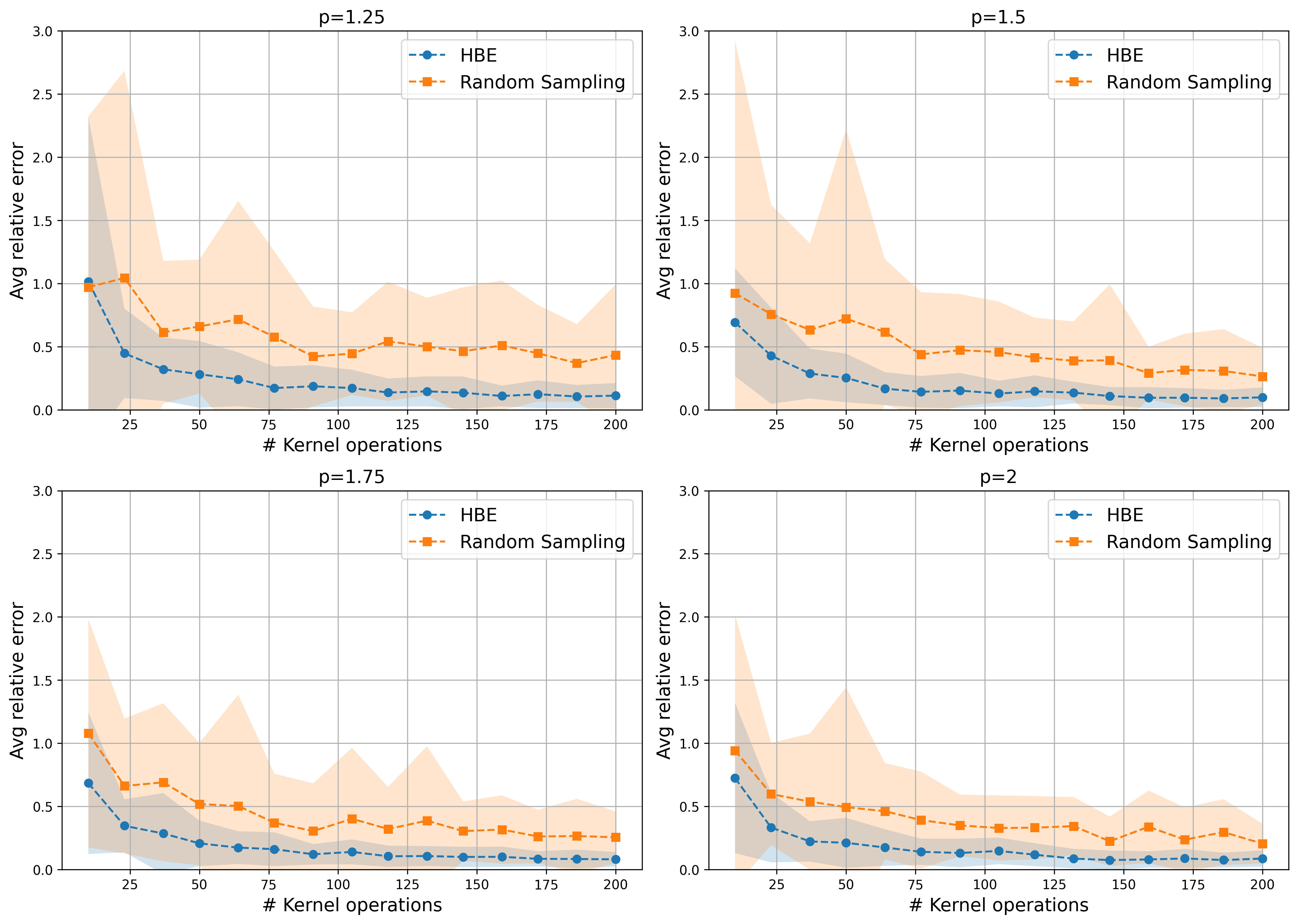}
    \caption{KDE via HBE results for Covertype}
    \label{fig:hbe_est_covertype}
\end{figure}

\textbf{(iii) KDE with relative error via HBE.}
We evaluate the HBE method \cite{charikar2002similarity,siminelakis2019rehashing,backurs2019space} with our mosaic LSH schemes compared to a uniformly random sampling baseline, which, as mentioned in \Cref{sec:kde}, achieves the same relative error guarantee for KDE with a larger asymptotic sample size.  
\Cref{fig:hbe_est_covertype} shows that the HBE method achieves lower error per number of kernel evaluations. The same empirical results were shown in \cite{siminelakis2019rehashing,backurs2019space} only for the 1-exponential and 2-exponential kernels.

\section{Concluding Remarks and Open Questions}\label{sec:conclusion}

We proved exact kernel LSH schemes for $p$-exponential kernels in bounded domains, including in the important case $p=2$ which was previously open. This substantially expands the available kernel LSH results in the literature. As applications, we obtained efficient LSHability-based KDE methods. 

Our techniques depart from previous LSHability techniques which were based on randomly shifted grids. 
Rather, we obtain our results by 
introducing the use of Poisson hyperplane processes into LSH. 
As such, our LSH schemes are not generalizations of the existing LSH schemes for $p=1$. In fact, our techniques work for all $p\in(1,2]$ but \emph{not} for $p=1$. The technical reason is that the 1-stable distribution (Cauchy) has no finite absolute first moment $m_p$, which our method requires.

Two notable limitations of our LSH schemes are the $O(d^2)$ dependence in their running time and their restriction to bounded regions. In comparison, existing LSH schemes for $p=1$ have running time $O(d)$ and they apply to all of $\R^d$, owing to the axis-alignment and translation-invariance of uniform grids in $\R^d$.  Our $\ell_1$-biased $p$-stable PHPs do not possess those properties. 
Whether these aspects of \Cref{thm:main} can be improved, particularly for $p=2$, is an interesting open problem.

\newpage
\section*{AI Use Statement}
During the course of this research, OpenAI's GPT 5.0 through 5.4 were used to assist in literature search, understanding mathematical background and exploring proof ideas. 
The text and the proofs in the paper were written without AI assistance. 
Anthropic's Claude Opus 4.7 was used for interactive coding for the experiments. 
 The core part of the code (the implementation of mosaic LSH) was written without AI assistance, and the code base as a whole was thoroughly verified. 
The authors are accountable for all of the text, proofs and artifacts associated with this paper. 

\begin{ack}
    This research was supported by the Israeli Ministry of Innovation, Science \& Technology, by Len
Blavatnik and the Blavatnik Family foundation, and by an Alon Scholarship.
\end{ack}

\bibliographystyle{plain}
\bibliography{pexplsh}


\newpage
\appendix

\section{Rejection Sampling for $p\in(1,2)$: Proof of \Cref{lmm:rejection}}\label{app:rejection}
Recall that in rejection sampling, we have a target density $f(x)$ we want to sample from and such that we can evaluate it at any given point $x$, and a proposal density $g(x)$ which we can sample from, such that they satisfy the envelope property $\forall x,\; f(x)\leq m\cdot g(x)$ for some $m>0$. The rejection sampler operates as follows:
\begin{enumerate}
    \item Sample $X\sim g$ and $U\sim\mathrm{Uniform}(0,1)$. 
    \item If $f(X)\geq mUg(X)$, accept $X$ as a sample from $f$ and terminate.
    \item Otherwise, reject the sample and return to step 1. 
\end{enumerate}
An accepted sample is provably distributed according to $f$. The expected number of iterations until acceptance is at most $m$. 

Our goal is to devise a rejection sampling scheme for $\bar\Theta$. 
Since the $\Theta$-marginal of $\mu$ is uniform, and $\bar\mu$ is proportional to $|F_p(\theta)w^{(p-1)/p}|\cdot\mu$, the $\bar\Theta$-marginal of $\bar\mu$ is proportional to $|F_p(\theta)|$. We thus need a rejection sampling scheme for sampling $\theta\sim(-\frac{\pi}{2},\frac{\pi}{2})$ with target density proportional to $|F_p(\theta)|$.

We start by calculating the normalization factor. Let $G_p(\theta)$ be defined over $\theta\in[0,\tfrac{\pi}2)$ as
\[
  G_p(\theta) = \cos((p-1)\theta)^{1/p}\cos(\theta)^{1-1/p} . 
\]
One can verify by direct derivation that $\tfrac{d}{d\theta}G_p(\theta)=-(1-\tfrac1p)F_p(\theta)$. Consequently, using the fact that $|F_p(\theta)|$ is symmetric about zero and equals to $F_p(\theta)$ at $\theta\in[0,\tfrac\pi2)$,
\[
  \int_{-\pi/2}^{\pi/2}|F_p(\theta)|d\theta
  = 2\int_{0}^{\pi/2}F_p(\theta)d\theta
  = 2\cdot\frac{-G_p(\pi/2)+G_p(0)}{1-1/p}
  = \frac{2}{1-1/p} .
\]
Therefore, the exact target density we need to sample from is
\[
f_{\bar\Theta}(\theta) = \tfrac12\left(1-\tfrac1p\right)|F_p(\theta)| . 
\]

We now define the proposal density $g_V(v)$ by the following sampling process: Sample $Y\sim\mathrm{Beta}(1-1/p,1)$ and a uniformly random sign $S\in\{1,-1\}$. Let $R=\mathrm{arccos}(Y)$. Return the sample $V=SR$.

We derive a closed form for $g_V(\theta)$. 
Since $Y\sim\mathrm{Beta}(\beta,1)$ takes values in $(0,1)$ with density $\beta y^{\beta-1}$, then $R$ takes values in $(0,\tfrac12\pi)$ and its density $f_R(\theta)$ is, by the change of variables formula for densities with $Y=\cos(R)$,
\[
  f_R(\theta) = \left(1-\tfrac1p\right)(\cos(\theta))^{-1/p}\left|\tfrac{d}{d\theta}\cos(\theta)\right|
  = \left(1-\tfrac1p\right)(\cos(\theta))^{-1/p}\sin(\theta) .
\]
Therefore, after the multiplication by the random sign $S$, we find that $V$ takes values in $(-\tfrac12\pi,\tfrac12\pi)$ with density  
\[
  g_V(\theta) = \tfrac12\left(1-\tfrac1p\right)(\cos(\theta))^{-1/p}|\sin(\theta)| .
\]
Finally, it remains to calculate the envelope parameter $m$. 
Since both the target density $f_{\bar\Theta}(\theta)$ and the proposal density $g_V(\theta)$ are symmetric about zero, it suffices to calculate $m$ for $\theta\in(0,\tfrac12\pi)$. In this range, for $p\in(1,2)$ we have $|\sin(p\theta)|=\sin(p\theta)$. Therefore, for all $\theta\in(0,\tfrac12\pi)$, 
\begin{align*}
    m_\theta := \frac{f_{\bar\Theta}(\theta)}{g_V(\theta)} &= \frac{\tfrac12\left(1-\tfrac1p\right)(\cos(\theta))^{-1/p}(\cos((1-p)\theta))^{-(1-1/p)}\sin(p\theta)}{\tfrac12\left(1-\tfrac1p\right)(\cos(\theta))^{-1/p}\sin(\theta)} \\
    &= \frac{\sin(p\theta)}{\sin(\theta)(\cos((1-p)\theta))^{1-1/p}}.
\end{align*}
Let $\theta'=(p-1)\theta$. Then the above becomes
\begin{align*}
    m_\theta &= \frac{\sin(\theta+\theta')}{\sin(\theta)(\cos(\theta')^{1-1/p}} & \\
    &\leq \frac{\sin(\theta+\theta')}{\sin\theta\cdot\cos\theta'} & (\cos(\theta'))^{1/p}\leq 1 \\
    &= \frac{\sin\theta\cdot\cos\theta' + \sin\theta'\cdot\cos\theta}{\sin\theta\cdot\cos\theta'} & \text{sum-product trigonometric identity} \\
    &= 1+\tan\theta'\cdot\cot\theta. & 
\end{align*}
Since the function $x\mapsto\tan(x)/x$ is increasing on $(0,\tfrac12\pi)$, and $\theta'\leq\theta$, we have
\[
  \tan\theta' \leq \frac{\theta'}{\theta}\cdot\tan\theta = (p-1)\tan\theta,
\]
and therefore the above yields
\[
  m_\theta \leq 1+(p-1)\cdot\tan\theta\cdot\cot\theta = p .
\]
Therefore, the envelope bound for rejection sampling with target density $f_{\bar\Theta}(\theta)$ and proposal density $g_V$ is $m\leq p$, and \Cref{lmm:rejection} is proven.

\section{Omitted Details from \Cref{sec:kde}: DP KDE}
\label{app:dpkde}

We recall the differential privacy definitions relevant to \Cref{thm:dpkde}.

\begin{definition}[differential privacy]
    Let $\epsilon>0$. Let $M$ be a randomized algorithm that maps datasets with elements in $\mathcal{X}$ to a space of outputs $\mathcal{O}$. $M$ is called an $\epsilon$-differentially private ($\epsilon$-DP) mechanism if for every two datasets $X,X'$ such that $X'$ is obtained from $X$ by either adding or dropping a single element, and for every set of outputs $O\subset\mathcal O$, it holds that $\Pr[M(X)\in O]\leq e^\epsilon\Pr[M(X')\in O]$.
\end{definition}

\begin{definition}[DP KDE in the function release model]
    Let $\epsilon>0$ and $\alpha,\eta\in(0,1)$. Let $\kk:\mathcal X\times\mathcal X\rightarrow[0,1]$ be a kernel. In the DP KDE problem in the function release model, the goal is to design an $\epsilon$-DP mechanism $M$ that given a dataset $X\subset\mathcal X$, returns a function $F_X:\mathcal{X}\rightarrow[0,1]$ that for each $y\in\mathcal X$ satisfies $\Pr[|F_X(y)-KDE_X(y)|\leq\alpha]\geq1-\eta$. 
\end{definition}

The mechanisms we obtain in \Cref{thm:dpkde} satisfy the above definitions. \Cref{thm:dpkde} follows from plugging \Cref{thm:main} into Theorem D.3 and Proposition D.4 in \cite{wagner2023fast}. The only thing to check is that the efficiency bounds in Proposition D.4 in \cite{wagner2023fast} depend on $\log B$, where $B$ is the number of bins of the LSH scheme. In our case, the number of bins is $2^N$ where $N\sim\Poi(2d/\sigma)$ is the number of hyperplanes sampled in the mosaic LSH scheme. 
Therefore, $\E[\log_2B]=\E[N]=2d/\sigma$.

\section{Non-LSHability of the Gaussian Kernel}\label{sec:nonlsh}
In this section we show a negative result, proving the lack of LSH schemes for a class of kernels that includes the Gaussian kernel, which is the $p=2$ case of the theorem below.

\begin{theorem}
    Let $p\geq1$ and $d\geq1$. Consider the kernel $\kk_p(x,y) = \exp(-\norm{x-y}_p^p)$ over $\R^d$. Then, $\kk_p$ is LSHable if and only if $p=1$.
\end{theorem}
\begin{proof}
    The $p=1$ is the 1-exponential kernel whose LSHability is known. We now assume that $p>1$ and show that $\kk$ is not LSHable.
    
    Charikar \cite{charikar2002similarity} showed that if $\kk: \mathcal{X}\times\mathcal{X} \to [0,1]$ is LSHable then $1 - \kk$  satisfies the triangle inequality. Therefore, if $\kk_p$ is LSHable then the following must hold $\forall a,b,c \in \mathcal{X}$ :
\begin{align*}
        & 1-\kk_p(a,b) + 1 - \kk_p(b,c) \geq 1 - \kk_p(a,c)
         \\
        &\iff  \kk_p(a,c) - \kk_p(a,b)- \kk_p(b,c) +1 \geq 0. \numberthis\label{eq:tri_ineq}
\end{align*}

    We construct a contradictory example where \cref{eq:tri_ineq} fails to hold and this will conclude our proof.
    It suffices to prove the claim for $d=1$ since a negative example for $\R$ embeds in $\R^d$ for all $d\geq 1$. 
    Take $\mathcal{X} = \R$ and $a=0, \; b=x, \; c=2x,$ as three evenly spaced points where $x\in \R^+$. 
    Then \cref{eq:tri_ineq} becomes 
    \begin{align*}
    & \kk_p(a,c) - \kk_p(a,b)- \kk_p(b,c) +1 \geq 0\\
    &\iff e^{-\norm{2x}_p^p} - e^{-\norm{x}_p^p} - e^{-\norm{x}_p^p} +1 \geq 0\\
    &\iff e^{-(2x)^p} -2e^{-x^p} + 1\geq 0. \numberthis\label{eq:tri_ineq_x}
    \end{align*}
    We denote the left-hand side of \cref{eq:tri_ineq_x} by $f_p(x) \coloneqq  e^{-(2x)^p} + 1 -2e^{-x^p} $. We show that there is $x\in\R^+$ such that $f_p(x)<0$. We apply a change of variable $u = x^p$
    and define
    \[ g(u) = e^{-2^pu} +1 -2e^{-u}, \]
    so we have $g(u) = f_p(x)$. At $u=0$ we have $g(0)=0$. 
    We calculate the derivative:
    $$g'(u) 
    = -2^pe^{-2^pu} + 2e^{-u}.$$
    At $u=0$,
    $$g'(0) = 2e^{-0} - 2^pe^{-0} = 2- 2^p  < 0,$$
    where the last inequality is since $p>1$. Thus we have $g(0)=0$ and $g'(0)<0$. Therefore, for sufficiently small $u^*>0$, we have $g(u^*) <0$, and thus for $x^* = (u^*)^{\frac{1}{p}} >0$ we get $f_p(x^*) <0 $.
\end{proof}

\section{Additional Experimental Details and Results}\label{app:additional_results}

\paragraph{Dataset parameters.} 
Covertype has 581K points and dimension 55. Census has 2.45M points and dimension 68. 
For efficiency we sample a subset of each dataset (note that subsampling preserves all KDE values). 
For Covertype we sample at rate 25\% and obtain 145K points. For Census we sample at rate 5\% and obtain 123K points. 

\paragraph{Bandwidth selection.}
In line with prior work, we choose two bandwidth settings for each kernel. 
One is chosen to render typical KDE values of order $10^{-2}$ for experiments measuring additive error. We use it for the kernel estimation and DP KDE experiments. 
The other bandwidth setting is chosen to render typical KDE values of order $10^{-3}$ for experiments measuring relative error. We use this setting for HBE experiments since as explained in \Cref{sec:kde}, the method is designed for regimes where KDE values we need to approximate are very small. 

For Covertype, the bandwidth settings are $0.15$ and $0.05$ respectively. For Census, the bandwidth settings are $0.42$ and $0.2$ respectively. Note that those values account for scaling the coordinates of each dataset into the range $[-1,1]$ to align with the normalization in \Cref{thm:main}.

\paragraph{Experimental parameters.} 
The kernel estimation experiments measure average absolute additive error over 1000 random pairs of points. 
The DP KDE experiments measure the average absolute additive error over 100 query points chosen at random from the dataset. 
The HBE experiments measure the average relative error over 100 query points chosen at random from the dataset. 
Error bars ($\pm$std) are shown as shaded regions around each curve.

\paragraph{Additional experimental results.}
The results for the Census datasets are shown in \Cref{fig:direct_kernel_est_census,fig:lsq_eps_est_census,fig:hbe_est_census}.

\begin{figure}
    \centering
    \includegraphics[width=0.8\textwidth]{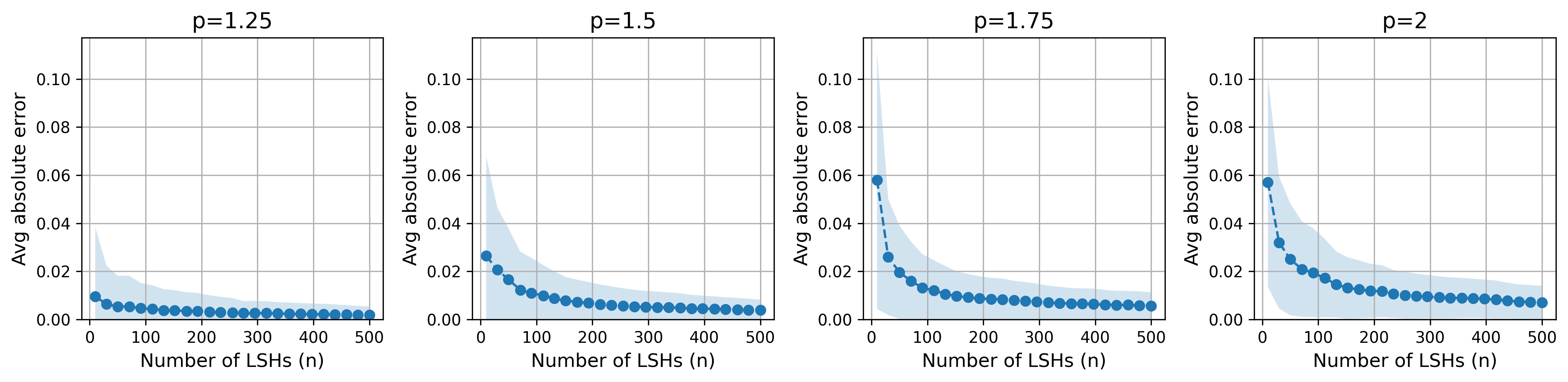}
    \caption{$p$-Exponential kernel estimation by mosaic LSH on the Census dataset}
\label{fig:direct_kernel_est_census}
\vskip 1em
    \includegraphics[width=0.8\textwidth]{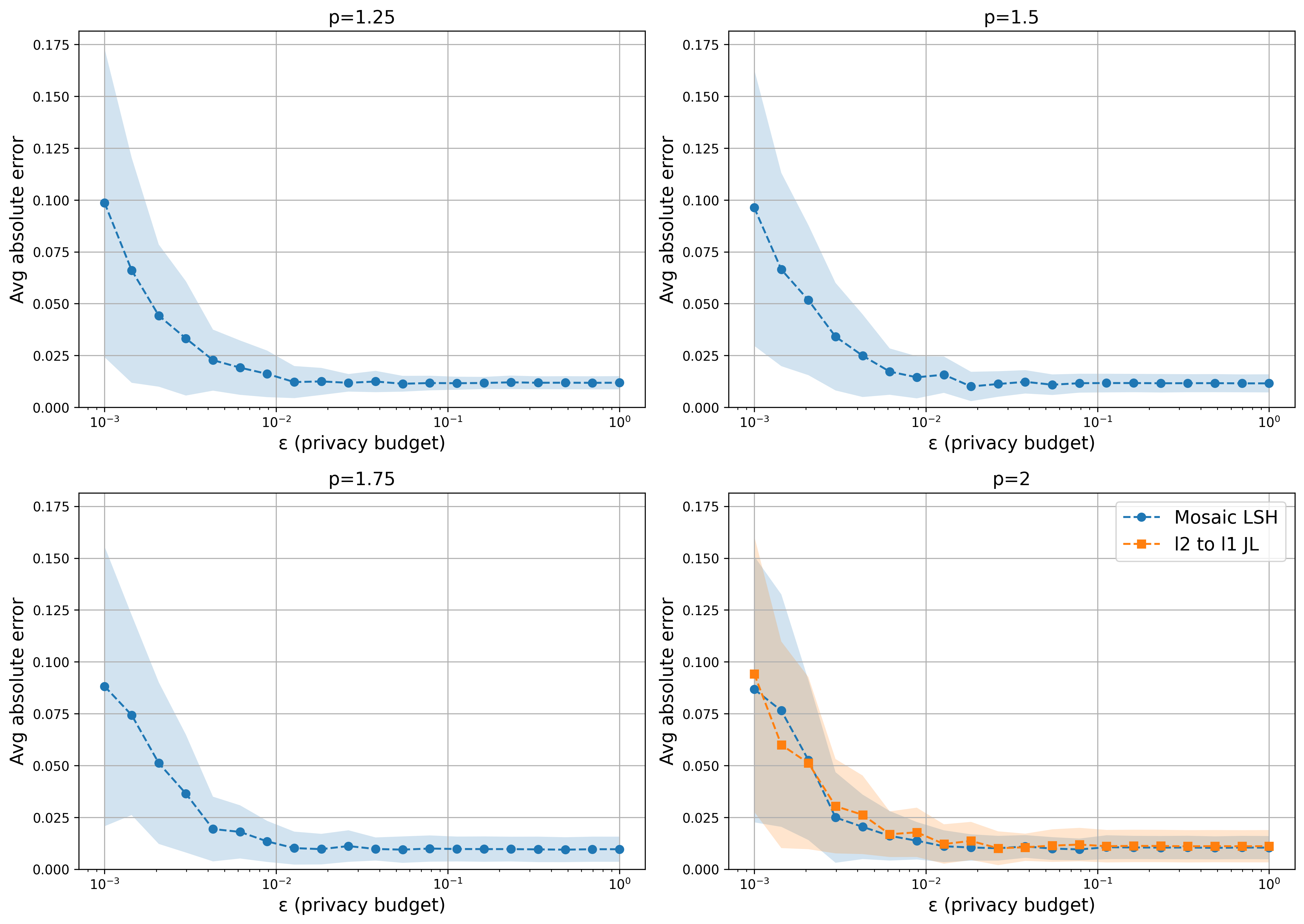}
    \caption{Differentially private KDE results for Census}

    \label{fig:lsq_eps_est_census}
    \vskip 1em
    \includegraphics[width=0.8\textwidth]{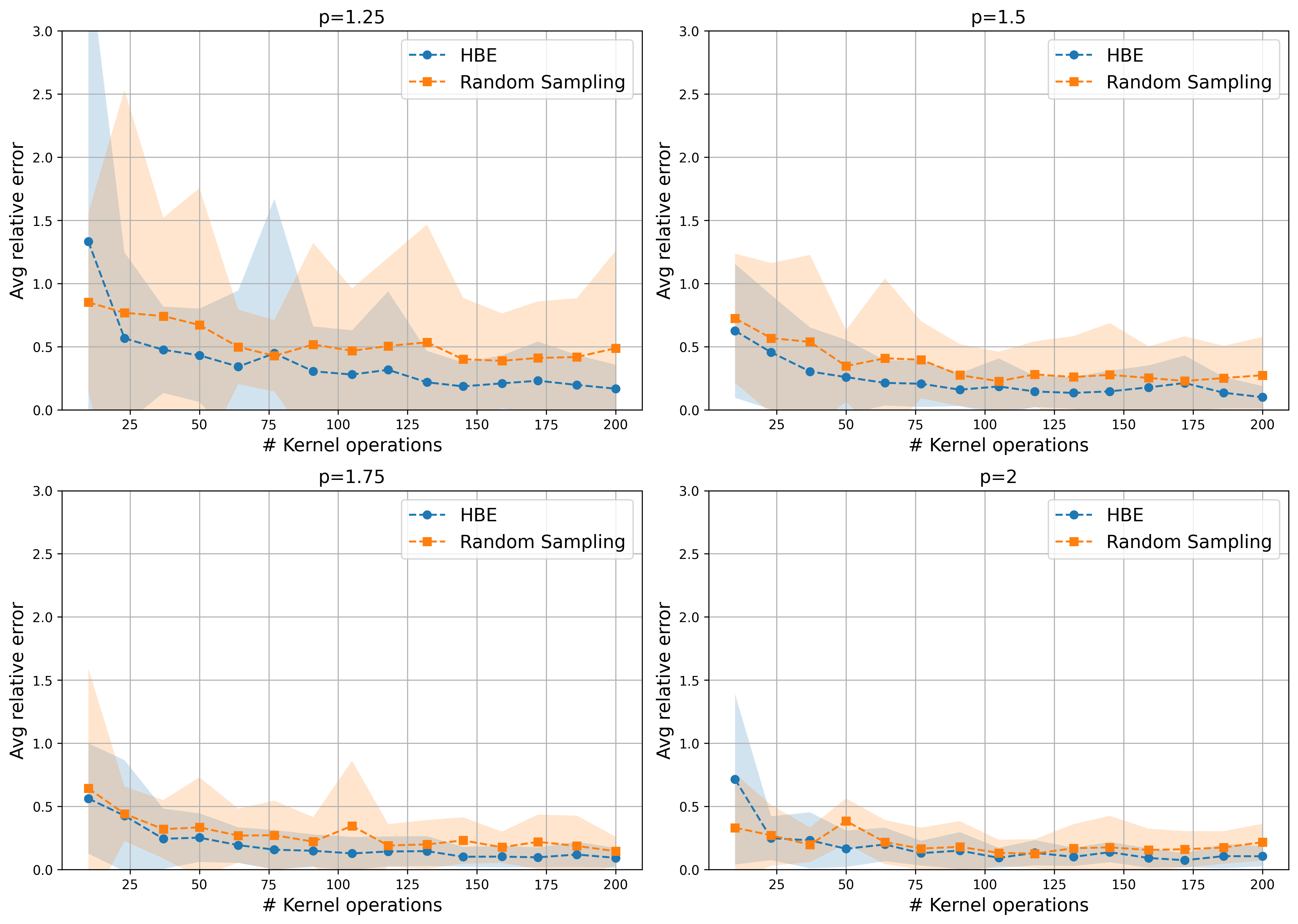}
    \caption{KDE via HBE results for Census}
    \label{fig:hbe_est_census}
\end{figure}

\end{document}